\documentclass{amsart}

\theoremstyle{definition}

\theoremstyle{remark}

\numberwithin{equation}{section}

\usepackage{amsmath,amssymb,amsrefs}
\usepackage{mathtools}
\usepackage{appendix}
\usepackage{pgfplots}
\usepackage[vlined,linesnumbered]{algorithm2e}

\pgfplotsset{compat=1.17}
\newtheorem{prop}{Proposition}

\newcommand{\HIDE}[1]{}

\newcommand{\ckh}[1]{#1}

\newcommand{\Det}[1]{\mathord{\mathord{\mathbf{det}}\left( #1 \right)}}
\newcommand{\phat}[1]{\mathord{\acute{#1}}}

\begin{document}

\title{A General Method for Generating Discrete Orthogonal Matrices}

\author{Ka-Hou~Chan*}
\address{School of Applied Sciences, Macao Polytechnic Institute, Macao, China}
\email{chankahou@ipm.edu.mo}
\thanks{*Corresponding Author}

\author{Wei~Ke}
\address{School of Applied Sciences, Macao Polytechnic Institute, Macao, China}
\email{wke@ipm.edu.mo}

\author{Sio-Kei~Im}
\address{Macao Polytechnic Institute, Macao, China}

\keywords{
	Discrete Orthogonal Matrices\and
	Discrete Cosine Transform\and
	Discrete Tchebichef Transform\and
	Orthogonal Polynomials\and
	Invertible Transformers
}

\date{\today}

\begin{abstract}
Discrete orthogonal matrices have several applications in information technology, such as in coding and cryptography.
It is often challenging to generate discrete orthogonal matrices.
A common approach widely in use is to discretize continuous orthogonal functions that have been discovered.
The need of certain continuous functions is restrictive.
To simplify the process while improving the efficiency and flexibility,
we present a general method for generating orthogonal matrices directly through the construction of certain even and odd polynomials
from a set of distinct positive values,
bypassing the need of continuous orthogonal functions.
We provide a constructive proof by induction that not only asserts the existence of such polynomials,
but also tells how to iteratively construct them.
Besides the derivation of the method as simple as a few nested loops,
we discuss two well-known discrete transforms, the Discrete Cosine Transform and the Discrete Tchebichef Transform.
How they can be achieved using our method with the specific values,
and show how to embed them into the transform module of video coding.
By the same token, we also show some examples of how to generate new orthogonal matrices from arbitrarily chosen values.
\end{abstract}

\maketitle

\section{Introduction}
\label{sec:intro}

Orthogonal transformations have very useful properties in solving science and engineering problems.
Just like the Fourier and Chebyshev series which are effective methods to project a periodic function into a series of linearly independent terms,
orthogonal polynomials provide a natural way to solve the related problems,
such as compression and protection in image processing~\cite{al2007combined,arman1993image,garg2019hybrid},
communication~\cite{wu2020self},
pattern recognition~\cite{monro2007dct,harakannanavar2019face},
feature capturing~\cite{dale2007dct,dhavale2012dwt} and
approximation theory\cite{arvesu2020multiple}.
Among various types of transformers, matrix transformers are most widely used
due essentially to their simplicity and explicitness,
especially for the transformations on real intervals ($\mathbb{R} \rightarrow \mathbb{R}$).
Even more important, orthogonal matrices are of a special type of transformers, for they are always invertible.
As a result, the source information can be recovered from the data that are transformed by an orthogonal matrix.

In image compression, most of the above mentioned applications and techniques deal with a lot of bulky source data,
such as video, audio and images, which often have real-time requirement.
Hence, data compression plays a major role in the storage and transmission.
Techniques such as the Discrete Cosine Transform~(DCT)~\cite{ahmed1974discrete}
is typically used in video encoding for transformations from the spatial domain to the frequency domain~\cite{song2000pvh},
followed by coding methods such as Huffman coding.
In recent years, the Discrete Tchebichef Transform~(DTT) provides another transformation method using
the Chebyshev moments~\cite{mukundan2003improving,mukundan2001image},
which has as good energy compression properties as the DCT and works better for
a certain class of images~\cite{hunt2004comparison}.
Both of the above example transformations are defined upon orthogonal polynomials.
The orthogonality is established over a continuous domain and approximated discretely over a certain amount of sample points.
Discrete orthogonal transformations have witnessed the interplay of signal processing,
semiconductor circuits, wireless networks and embedded systems to provide viable and cutting-edge technologies
that are truly the state-of-the-art.
The challenge lies in delivering practically realizable and economic solutions, while retaining the quality.

It is well-known that the orthogonality of two polynomials $P_i(x)$ and $P'_j(x)$, respectively having degrees $i$ and $j$,
is defined by extending the dot product of two vectors,
the sum of the products of the corresponding components,
to the integral of product $P_i(x) P'_j(x)$ over a continuous domain.
Formally, when the integral becomes zero, the two polynomials are orthogonal to each other, i.e.,
\begin{equation*}
	\int P_i(x) P'_j(x) dx = 0.
\end{equation*}
In practical applications, this definition is often approximated over a set of discrete samples $x_0, \dots, x_{n-1}$,
\begin{equation*}
	\sum\limits_{k=0}^{n-1} \left[ P_i(x_k) P'_j(x_k) \right] = 0.
\end{equation*}
Those satisfying this property are called discrete orthogonal polynomials~\cite{abramowitz1972handbook}.
Therefore, together with the degrees of polynomials ranging from $0$ to $n-1$, an $n \times n$ discrete orthogonal matrix
$\begin{bmatrix}
	P_i(x_k)
\end{bmatrix}$
can be constructed,
where any two different row vectors are orthogonal.
Discrete orthogonal matrices are commonly used in a number of orthogonal transformations over real intervals, such as
the Chebyshev Polynomials~\cite{weisstein2003chebyshev,mason2002chebyshev},
the Legendre Polynomials~\cite{weisstein2002legendre,sofonea2020inequality},
the Discrete Hartley Transform~\cite{bracewell1983discrete} and
the well-known Discrete Cosine Transform~\cite{ahmed1974discrete}.
A comprehensive overview of these orthogonal polynomials, along with the development of their discrete matrices, is also detailed in~\cite{roy2011}.

The purpose of this paper is to derive discrete orthogonal matrices directly by solving systems of linear equations,
rather than to discretize existing continuous orthogonal polynomials.
Our method has several advantages.
It has virtually no precondition to use.
Orthogonal matrices of arbitrary sizes can be generated to the need of an application.
It directly follows the definition of discrete orthogonality,
eliminating the need to discuss the orthogonal property over a continuous domain,
such as the interval $[-1, +1]$ of the Chebyshev Polynomials~\cite{mason2002chebyshev}.
The method focuses on how to derive the coefficients of the polynomials
that must be discretely orthogonal to each other over a set of given sample values,
for example,
$x_k = \cos\frac{\pi(2k+1)}{2n}$,
for $k=0,\dots,n-1$, are the values of the $n \times n$ orthogonal matrix for the Discrete Cosine Transform~(DCT)~\cite{ahmed1974discrete}.
The errors in the discretization of continuous functions can also be avoided.

Generating orthogonal matrices directly from a set of values give engineers a new way of obtaining such matrices with unlimited variations,
without the need to discover and prove the properties of orthogonal polynomials mathematically in the first place.
Although, by jumping to the construction directly, we sacrifice some mathematical insights and certainties,
we provide a way to significantly broaden the base of discrete orthogonal matrices for engineering analyses.
Our method is also simple and intuitive.
It starts with the definition of discrete orthogonality,
makes use of even and odd functions, inspired by the DCT and DTT, to simplify the problems,
constructs the linear equation system for deriving the coefficients of the polynomials,
proves that a unique solution exists, and finally
inductively obtains the solution.
Through the practicing of this method, we easily and effectively reproduce the orthogonal matrices for DCT and DTT in only a few simple steps.
We also generate a couple of others to show the potential and flexibility.

The rest of the paper is organized as follows.
Section~\ref{sec:related} goes through the related work.
Section~\ref{sec:orthogonal} presents the technical details and justifications of the orthogonal matrix generation method.
Section~\ref{sec:practice} reproduces a few well-known orthogonal matrices to show the effectiveness of the method.
Finally, Section~\ref{sec:conclude} concludes the paper.

\section{Related Work}
\label{sec:related}

In the past three decades, many researchers aimed to generalize the theory of
how to construct orthogonal polynomials of a single discrete variable,
as the solution of hypergeometric type differential equations, to that of multiple variables.
In early days, a method was designed in~\cite{markett1994linearization}
that began with the three-term recurrence relation for symmetric orthogonal polynomial systems
to set up a partial differential equation for the orthogonal polynomials,
in case of the connection problem, or
for the product of two orthogonal polynomials,
in case of the linearization problem.
This equation had to be solved in terms of the initial data to expand the coefficients.
In~\cite{van1999properties},
to make the relevant orthogonality measures continuous,
the parameter domain was carefully chosen.
This method focuses on a different way to obtain parameters,
where the orthogonality measure becomes merely discrete that it is finitely supported on the grid points with given weights.
Other discrete multi-variable extensions of hypergeometric orthogonal polynomials were considered in~\cite{xu2004discrete}.
Later, a novel set of discrete and continuous orthogonal matrices based on orthogonal polynomials was introduced into
the field of orthogonal polynomial generation~\cite{xu2005second}.
In~\cite{rodal2005orthogonal}, several relations linking the differences
between bivariate discrete orthogonal polynomials and general polynomials were given.
They presented a multi-variable generalization for all the discrete families,
that gave each family a hypergeometric representation and a orthogonality weight function,
proving that these polynomials were orthogonal with respect to the subspace of lower degrees
and biorthogonal within a given subspace~\cite{koekoek2010hypergeometric}.
Next, a systematic study of the orthogonal polynomial solutions to a second order partial differential equation
with two variables of hypergeometric type was made in~\cite{rodal2008structure}.
In the bivariate discrete case, a hypergeometric formula was also given in~\cite{rodal2007linear}.
For an infinitely differentiable function, the formula for the expansion coefficients of a general order derivative was available for the expansions in Chebyshev polynomials~\cite{doha2006recurrence}.
Thus the generation of recurrence relations to expand the coefficients of multi-variable orthogonal polynomials is similar to
that in the single, both continuous and discrete, variable case~\cite{ahmed2009recurrence}.
These results motivated those researchers interested in multidimensional mathematical physics problems
to use expansions in terms of orthogonal polynomials of multiple discrete variables.

Meanwhile, there were attempts in order to expand the coefficients of an arbitrary polynomial of a discrete variable
and evaluate the expanded coefficients of an orthogonal matrix.
Few advantages were achieved in these problems until
\cite{godoy1997minimal} presented the recent recursive approach,
and \cite{zarzo1997results} gave an alternative way to the approaches for producing classical orthogonal polynomials.
\cite{ronveaux1995recurrence} used recurrent equations to prove the positivity of the connection coefficients between certain instances of orthogonal polynomials.
They designed a constructive algorithm which allowed us to calculate recurrently the expansion coefficients of the evaluation problem.
However, this approach requires the knowledge of the differential equation of the polynomial to expand and
the recursion relation as well as the differential-difference relation must be prepared for the polynomials conforming the orthogonal set.
A few years later, the approach in~\cite{alvarez2000linearization}
presented a very similar algorithm for finding the recurrence relation for both the connection and linearization coefficients.
Also, another algorithm was developed for solving the connection problem between the four families of classical orthogonal polynomials~\cite{wozny2003recurrence}.
Recently, in~\cite{cohl2013generalization} there provides a series rearrangement technique
combining a connection relation with a generating function, resulting in a series with multiple sums.
Then, \cite{cohl2013generalizations} extended this technique to many generating functions to derive a generalized generating function whose coefficients were given in hypergeometric functions.
To the best of our knowledge, the coefficients of polynomials are always related to the polynomials of lower degrees
when they are in a series of orthogonal polynomials of the same type.
It leads to that the coefficients of a higher degree polynomial can be determined by some recursion or iteration relation of the corresponding linearization.
The order of summations are then rearranged and it is often simplified to the production of a generating function
whose coefficients are given in terms of the general or fundamental hypergeometric functions.

\section{Discrete Orthogonal Polynomials and Matrices}
\label{sec:orthogonal}

An $n\times n$ matrix $M$ is an orthogonal matrix if the transpose $M^T$ equals to the inverse $M^{-1}$.
Thus, an orthogonal matrix is always invertible.
By the definition $M M^T = I$, where $I$ is the identity matrix,
the rows of an orthogonal matrix form an orthonormal basis that
each row vector has length one, and is perpendicular to each other rows.
Formally speaking, the dot product of two row vectors $\vec{a}_i \cdot \vec{a}_j$ is $1$ when $i = j$, or $0$ otherwise, that is,
\begin{equation}\label{eq:condition}
	\sum\limits_{k=0}^{n-1} (a_{ik} \times a_{jk}) =
	\begin{cases}
		0, & i \neq j \\
		1, & i = j
	\end{cases}
\end{equation}
for $0 \leq i, j \leq n-1$.

We consider the construction of a type of orthogonal matrices from a set of values $x_0, x_1, \dots, x_{n-1}$, and
a set of polynomials $P_0(x), \dots, P_{n-1}(x)$ respectively of degrees from $0$ to $n-1$.
We denote the coefficients of the polynomial expansions by $c_{(i,j)}$, such that
$$
P_i(x) = \sum_{k=0}^{i} c_{(i,k)} x^{i-k},
$$
for $0 \leq i \leq n-1$. We then construct the orthogonal matrix of the form
\begin{equation}
	\label{eq:orthogonal}
	M =
	\begin{bmatrix}
		P_i(x_k)
	\end{bmatrix}
	_{0 \leq i, k \leq n-1},
\end{equation}
by deriving the polynomials $P_0(x), \dots, P_{n-1}(x)$~\cite{day2005roots}.
These polynomials are called the orthonormal basis of the orthogonal matrix $M$~\cite{abramowitz1972handbook,chihara2011introduction,nikiforov1991classical}.

Together with the condition of orthogonal matrices in~(\ref{eq:condition}), we require
\begin{equation}
	\label{eq:poly-cond}
	\sum_{k=0}^{n-1} \left[ P_i(x_k) P_j(x_k) \right] =
	\begin{cases}
		0, & i \neq j \\
		1, & i = j
	\end{cases}
\end{equation}
for $0 \leq i, j \leq n-1$. An easy way to make a summation zero is to set half of the items the opposite values of the other half, for example, when $n = 2m$,
we should have
\begin{equation*}
	P_i(x_k) P_j(x_k) = - P_i(x_{k+m}) P_j(x_{k+m}),
\end{equation*}
for $0 \leq i, j \leq 2m-1$ and $0 \leq k \leq m-1$. We can further refine this condition to
\begin{equation}
	\label{eq:cancel}
	P_i(x_k) = -P_i(x_{k+m}) \quad \text{and} \quad P_j(x_k) = P_j(x_{k+m}).
\end{equation}
It's clear that when $x_k = -x_{k+m}$, the condition in~(\ref{eq:cancel}) can be fulfilled if $P_i(x)$ is an odd function and $P_j(x)$ an even function.

Based on the analysis, we narrow the range of the polynomials down to only even and odd functions,
together with a set of opposite values to make use of the parity as above.
Given $m$ distinct values $y_0, \dots, y_{m-1} > 0$, we choose $\pm y_0, \dots, \pm y_{m-1}$ as the set of values for the matrix construction.
Thus, the matrix in~(\ref{eq:orthogonal}) is formulated as
\begin{equation}\label{eq:sym-orthogonal}
	\begin{bmatrix}
		P_i(-y_0) & \cdots & P_i(-y_{m-1}) & P_i(+y_0) & \cdots & P_i(+y_{m-1})
	\end{bmatrix},
\end{equation}
for $0 \leq i \leq 2m-1$. We are going to derive the orthogonal matrix in (\ref{eq:sym-orthogonal}) by resolving the coefficients of polynomials $P_0, \dots, P_{2m-1}$
based on the set of values $\pm y_0, \dots, \pm y_{m-1}$.

\subsection{Even and Odd Polynomials}

Consider the expansion of an $i$-degree polynomial.
When $i = 2t$, an even polynomial can be constructed by removing all the odd-degree terms.
Thus, the expansion of such an even polynomial can be written as
\begin{equation}
	\label{eq:even-poly}
	P_{2t}(x) = \sum_{p=0}^{t} \left[ c_{(2t,2p)} x^{2(t-p)} \right].
\end{equation}
Similarly, when $i = 2t+1$, an odd polynomial can be obtained by multiplying an $x$ to every and each term in~(\ref{eq:even-poly}), where all the even-degree terms are removed,
\ckh{
	\begin{equation}
		\label{eq:odd-poly}
		P_{2t+1}(x) = \sum_{p=0}^{t} \left[ c_{(2t+1,2p+1)} x^{2(t-p)+1} \right].
	\end{equation}
}
For the parity properties of even and odd polynomials, we have
\begin{equation}\label{eq:parity}
	P_{2t}(-x) = P_{2t}(x) \quad \text{and} \quad P_{2t+1}(-x) = -P_{2t+1}(x).
\end{equation}
Now, we limit the choice of the $P_i$ polynomials to those of the forms in~(\ref{eq:even-poly}) and~(\ref{eq:odd-poly}).
The number of the unknown coefficients is reduced to $t+1$ for each of the $(2t)$- and $(2t+1)$-degree polynomials.
We are going to derive these unknown coefficients based on the condition of orthogonal matrices in~(\ref{eq:poly-cond}).
We substitute the rows of the matrix in~(\ref{eq:sym-orthogonal}) for the rows $P_i$ and $P_j$ in condition~(\ref{eq:poly-cond}),
\begin{equation*}
	\sum_{k=0}^{2m-1} \left[ P_i(x_k) P_j(x_k) \right]
	= \sum_{k=0}^{m-1} \left[ P_i(-y_k) P_j(-y_k) + P_i(+y_k) P_j(+y_k) \right]
	=
	\begin{cases}
		0, & i \neq j \\
		1, & i = j
	\end{cases} \\
\end{equation*}
for $0 \leq i,j \leq 2m-1$, and consider the parity property in (\ref{eq:parity}), we have
\begin{equation}\label{eq:condition2}
	\sum_{k=0}^{2m-1} \left[ P_i(x_k) P_j(x_k) \right]
	=\begin{cases}
		0, & i \not\equiv j \pmod{2} \\
		2 \times \sum_{k=0}^{m-1} \left[ P_i(y_k) P_j(y_k) \right], & i \equiv j \pmod{2}.
	\end{cases}
\end{equation}
To derive the coefficients for the orthogonal matrix, we focus on the case of $i \equiv j \pmod{2}$,
where the sum is required to be $1$ when $i = j$, or $0$ otherwise.

\subsection{Polynomial Coefficient Induction}
\label{ssec:coef-induct}

The dot product of a row in an orthogonal matrix with itself is $1$, or $0$ with another row.
An even polynomial $P_{2t}(x)$ in (\ref{eq:even-poly}) has only $t+1$ coefficients to resolve.
\ckh{
	If we take the highest coefficient (with $p=0$) out and resolve it later by the unit length condition,
	there are only $t$ coefficients left,
	\begin{equation*}
		d_{(2t,2p)} = \frac{c_{(2t,2p)}}{c_{(2t,0)}} \quad\text{and}\quad d_{(2t+1,2p+1)} = \frac{c_{(2t+1,2p+1)}}{c_{(2t+1,1)}},
	\end{equation*}
	for $1 \leq p \leq t$, and we have $d_{(2t,0)} = d_{(2t+1,1)} = 1$.
	We denote this form of $P(x)$ as $\hat{P}(x)$, i.e., $\hat{P}_{2t}(x) = \frac{1}{c_{(2t,0)}} P_{2i}(x)$ and $\hat{P}_{2t+1}(x) = \frac{1}{c_{(2t+1,1)}} P_{2i+1}(x)$ respectively.
}

Obviously, we can safely replace those $P(x)$ with $\hat{P}(x)$
in the discussion of obtaining the perpendicularity between two rows in the matrix, since there is only a scalar difference.
There are exactly $t$ even polynomials, $\hat{P}_0(x), \hat{P}_2(x), \dots, \hat{P}_{2(t-1)}(x)$ with smaller degrees in the matrix.
By the condition that the $t$ rows constructed by these smaller polynomials are perpendicular to the row from polynomial $\hat{P}_{2t}(x)$,
it establishes a system of $t$ equations.
If there are solutions to the equation system,
and we can find a general way to solve the coefficients $d_{(2t,2p)}$, for $1 \leq p \leq t$, from these equations,
then we are able to obtain the coefficients of all the polynomials from $\hat{P}_0(x)$ to $\hat{P}_{2t}(x)$ inductively.
The base case is trivial, that is, $\hat{P}_0(x) = 1$.

For such a matrix of size $2m \times 2m$, the equation system for the coefficients of $\hat{P}_{2t}(x)$ is straightforward,
by letting the dot products with those smaller even polynomials be $0$,
\begin{equation*}
	\sum_{k=0}^{m-1} \left[ \hat{P}_{2i}(y_k) \hat{P}_{2t}(y_k) \right]
	= \sum_{k=0}^{m-1} \left[ \hat{P}_{2i}(y_k) \left(y_k^{2t}+\sum_{p=1}^t \left( d_{(2t,2p)} y_k^{2(t-p)} \right)\right) \right] = 0,
\end{equation*}
for $0 \leq i \leq t-1$.
Then, we examine the terms containing a certain coefficient $d_{(2t,2p)}$, for $1 \leq p \leq t$. The above equation system can be written as
\begin{equation*}
	\sum_{k=0}^{m-1} \left[ \hat{P}_{2i}(y_k) y_k^{2t} \right] +
	\sum_{p=1}^t
	\sum_{k=0}^{m-1} \left[ d_{(2t,2p)} \left( \hat{P}_{2i}(y_k) y_k^{2(t-p)} \right) \right] = 0,
\end{equation*}
for $0 \leq i \leq t-1$. Thus, we have a linear equation system for the unknown coefficients as
\begin{equation}
	\label{eq:low-coef}
	A_t D_t = -B_t,
\end{equation}
where
\begin{equation*}
	%
	\begin{aligned}
		A_t =& \mathop{\begin{bmatrix}
				\sum\limits_{k=0}^{m-1} \left( \hat{P}_{2i}(y_k) y_k^{2(t-p)} \right)
		\end{bmatrix}}
		\limits_{0 \leq i \leq t-1,1 \leq p \leq t},\\
		D_t =& \mathop{\begin{bmatrix}
				\vphantom{\sum\limits_{k=0}^{m-1}}
				d_{(2t,2p)}
		\end{bmatrix}}
		\limits_{1 \leq p \leq t},\\
		B_t =& \mathop{\begin{bmatrix}
				\sum\limits_{k=0}^{m-1} \left( \hat{P}_{2i}(y_k) y_k^{2t} \right)
		\end{bmatrix}}
		\limits_{0 \leq i \leq t-1}.
	\end{aligned}
\end{equation*}
We induct on $t$ to prove that the determinant $\Det{A_t} \neq 0$, thus (\ref{eq:low-coef}) has a unique solution to $D_t$.
\begin{prop}
	\label{prop:det-nz}
	For $1 \leq t \leq m-1$, $\Det{A_t} \neq 0$.
\end{prop}
\begin{proof}
	The base case is trivial that $A_1 = \begin{bmatrix} m \end{bmatrix}$, thus $\Det{A_1} = m \neq 0$.
	
	When $2 \leq t \leq m-1$, we have
	\begin{equation*}
		\hat{P}_{2(t-1)}(y_k) = \mathop{\begin{bmatrix}
				y_k^{2(t-p)}
			\end{bmatrix}^T}
		\limits_{\mathclap{1 \le p \le t}}
		\begin{bmatrix}
			1 \\
			D_{t-1}
		\end{bmatrix}
		\quad
		(0 \leq k \leq m-1).
	\end{equation*}
	By induction hypothesis, $D_{t-1}$ has a unique solution, also by the Cramer's rule,
	\begin{equation*}
		D_{t-1} = \begin{bmatrix}
			\dfrac
			{\Det{A_{t-1} \left[B_{t-1} \middle/ p \right]}}
			{\Det{A_{t-1}}}
		\end{bmatrix}
		_{1 \le p \le t-1},
	\end{equation*}
	where $A[B/p]$ is the matrix formed by replacing the $p$-th column of $A$ by the column vector $B$. Consider the matrix
	\begin{equation*}
		C_t(x) = \begin{bmatrix}
			\begin{array}{c|@{\hspace{1.5em}}c@{\hspace{1.5em}}}
				B_{t-1} & A_{t-1} \\[1ex]
				\hline
				\multicolumn{2}{c}{} \\[-2ex]
				\multicolumn{2}{c}{\begin{bmatrix}
						x^{2(t-p)}
					\end{bmatrix}^T
					_{1 \le p \le t}}
			\end{array}
		\end{bmatrix},
	\end{equation*}
	and the cofactor expansion of $\Det{C_t(y_k)}$ along the bottom row, we establish the following identity,
	\begin{equation*}
		\Det{C_t(y_k)} = \Det{A_{t-1}} \times \hat{P}_{2(t-1)}(y_k),
	\end{equation*}
	for $0 \leq k \leq m-1$. Furthermore, if we partition $A_t$ similarly, then we get
	\begin{equation*}
		A_t = \begin{bmatrix}
			\begin{array}{c|@{\hspace{1.5em}}c@{\hspace{1.5em}}}
				B_{t-1} & A_{t-1} \\[1ex]
				\hline
				\multicolumn{2}{c}{} \\[-2ex]
				\multicolumn{2}{c}{\begin{bmatrix}
						\sum\limits_{k=0}^{m-1}\left( \hat{P}_{2(t-1)}(y_k) \times y_k^{2(t-p)} \right)
					\end{bmatrix}^T
					_{1 \le p \le t}}
			\end{array}
		\end{bmatrix}.
	\end{equation*}
	Thus, by comparing $A_t$ with $C_t$, we have the following conclusion,
	\begin{equation*}
		\Det{A_t} = \sum_{k=0}^{m-1}\left[ \hat{P}_{2(t-1)}(y_k) \times \Det{C_t(y_k)} \right] = \Det{A_{t-1}} \times \sum_{k=0}^{m-1} \left[ \hat{P}_{2(t-1)}(y_k) \right]^2 \neq 0.
	\end{equation*}
	Notice that $\hat{P}_{2(t-1)}$ is an even function, thus it has at most $t-1$ positive roots.
	On the other hand, we have $m$ distinct positive $y_k$ values and $t \leq m-1$,
	therefore the sum of the squares above cannot be zero.
\end{proof}

\ckh{
	For those odd polynomials $\hat{P}_{2t+1}(x)$ ($0 \leq t \leq m-1$), the base case is $\hat{P}_1(x) = x$.
}
We can obtain an equation system similar to (\ref{eq:low-coef}) to solve the coefficients inductively.
We denote this equation system as
\begin{equation}
	\label{eq:low-coef-odd}
	\phat{A}_t \phat{D}_t = -\phat{B}_t,
\end{equation}
where
\begin{equation*}
	%
	\begin{aligned}
		\phat{A}_t =& \mathop{\begin{bmatrix}
				\sum\limits_{k=0}^{m-1} \left( \hat{P}_{2i+1}(y_k) y_k^{2(t-p)+1} \right)
		\end{bmatrix}}
		\limits_{0 \leq i \leq t-1,1 \leq p \leq t},\\
		\ckh{
			\phat{D}_t =& \mathop{\begin{bmatrix}
					\vphantom{\sum\limits_{k=0}^{m-1}}
					d_{(2t+1,2p+1)}
			\end{bmatrix}}
			\limits_{1 \leq p \leq t}
		},\\
		\phat{B}_t =& \mathop{\begin{bmatrix}
				\sum\limits_{k=0}^{m-1} \left( \hat{P}_{2i+1}(y_k) y_k^{2t+1} \right)
		\end{bmatrix}}
		\limits_{0 \leq i \leq t-1}.
	\end{aligned}
\end{equation*}
We can also prove that (\ref{eq:low-coef-odd}) has a unique solution to $\phat{D}_t$.
\begin{prop}
	\label{prop:det-nz-odd}
	For $1 \leq t \leq m-1$, $\Det{\phat{A}_t} \neq 0$.
\end{prop}
\begin{proof}
	This proof is almost identical to the proof of Proposition~\ref{prop:det-nz}, with a different base case $\phat{A}_1$, where
	\begin{equation*}
		\phat{A}_1 = \begin{bmatrix}
			\sum\limits_{k=0}^{m-1}\left( \hat{P}_1(y_k)y_k \right)
		\end{bmatrix}
		= \begin{bmatrix}
			\sum\limits_{k=0}^{m-1} y_k^2
		\end{bmatrix}.
	\end{equation*}
	Notice that $\hat{P}_1(x) = x$. Since all $y_k > 0$, certainly we have $\Det{\phat{A}_1} \neq 0$.
	For the induction step, we substitute in the $\phat{A}$, $\phat{D}$ and $\phat{B}$ counterparts, together with
	\begin{equation*}
		\phat{C}_t(x) = \begin{bmatrix}
			\begin{array}{c|@{\hspace{1.5em}}c@{\hspace{1.5em}}}
				\phat{B}_{t-1} & \phat{A}_{t-1} \\[1ex]
				\hline
				\multicolumn{2}{c}{} \\[-2ex]
				\multicolumn{2}{c}{\substack{
						\\
						\begin{bmatrix}
							x^{2(t-p)+1}
						\end{bmatrix}^T
						_{1 \le p \le t}}}
			\end{array}
		\end{bmatrix}.
	\end{equation*}
	Also, for the number of positive roots of an odd polynomial $\hat{P}_{2(t-1)+1}$, we still have at most $t-1$,
	because zero is a root for any odd polynomial.
\end{proof}

The proofs also give us a method to derive the polynomials $\hat{P}_{2t}(x)$ and $\hat{P}_{2t+1}(x)$ inductively.
We have
\begin{equation*}
	\hat{P}_{2t}(x) = \mathop{\begin{bmatrix}
			x^{2(t-p)}
		\end{bmatrix}^T}
	\limits_{0 \leq p \leq t}
	\begin{bmatrix}
		1 \\
		D_t
	\end{bmatrix},
\end{equation*}
\begin{equation*}
	\hat{P}_{2t+1}(x) = \mathop{\begin{bmatrix}
			x^{2(t-p)+1}
		\end{bmatrix}^T }
	\limits_{0 \leq p \leq t}
	\begin{bmatrix}
		1 \\
		\phat{D}_t
	\end{bmatrix},
\end{equation*}
for $0 \leq t \leq m-1$,
where $D_t = A_t^{-1}B_t$ and $\phat{D}_t = \phat{A}_t^{-1}\phat{B}_t$ respectively.

\subsection{Obtaining Unit Vectors}
\label{ssec:algo}

To make each row vector of (\ref{eq:sym-orthogonal}) having the unit length, we refer to the condition in (\ref{eq:condition2}),
\begin{equation*}
	2 \times \sum_{k=0}^{m-1} \left[ P_i(y_k) \right]^2 = 1.
\end{equation*}
Thus, together with the fact $c_{(i,0)}\hat{P}_i(x) = P_i(x)$, we have
\begin{equation}
	\label{eq:first-coef}
	2 \times c_{(i,0)}^2 \sum_{k=0}^{m-1} \left[ \hat{P}_i(y_k) \right]^2 = 1\\\implies
	c_{(i,0)} = \pm \left( 2 \times \sum_{k=0}^{m-1} \left[ \hat{P}_i(y_k) \right]^2 \right)^{-\frac{1}{2}}.
\end{equation}
As a result, we have derived the method to obtain a $2m \times 2m$ orthogonal matrix
based on any set of $m$ distinct positive values. Algorithm~\ref{algo} presents the overall procedure.

\begin{algorithm}
	\KwData{$m$ distinct positive values $y_0,\dots,y_{m-1}$}
	\KwResult{a $2m\times2m$ orthogonal matrix $M$}
	\Begin{
		\For{$k \gets 0$ \KwTo $m-1$}{
			$\hat{P}_{0,k} \gets 1$; $\hat{P}_{1,k} \gets y_k$\label{algo:P1-coef}
		}
		\For{$t \gets 1$ \KwTo $m-1$}{
			\For{$i \gets 0$ \KwTo $t-1$}{
				\For{$p \gets 1$ \KwTo $t$}{
					$A_{i,p-1} \gets \sum_{k=0}^{m-1} \left[ \hat{P}_{2i,k} \times y_k^{2(t-p)} \right]$; $\phat{A}_{i,p-1} \gets \sum_{k=0}^{m-1} \left[ \hat{P}_{2i+1,k} \times y_k^{2(t-p)+1} \right]$
				}
				$B_i \gets \sum_{k=0}^{m-1} \left[ \hat{P}_{2i,k} \times y_k^{2t} \right]$;     $\phat{B}_i \gets \sum_{k=0}^{m-1} \left[ \hat{P}_{2i+1,k} \times y_k^{2t+1} \right]$
			}
			$D \gets A^{-1}\left(-B\right)$; $\phat{D} \gets \phat{A}^{-1}\left(-\phat{B}\right)$ \\
			\For{$k\gets0$ \KwTo $m-1$}{
				$\hat{P}_{2t,k} \gets y_k^{2t} + \sum_{p=1}^t \left[ y_k^{2(t-p)} \times D_{p-1} \right]$; $\hat{P}_{2t+1,k} \gets y_k^{2t+1} + \sum_{p=1}^t \left[y_k^{2(t-p)+1} \times \phat{D}_{p-1} \right]$ \\
			}
		}
		\For{$t \gets 0$ \KwTo $m-1$\label{algo:unit-vector}}{
			$c \gets \pm \left( 2 \sum_{k=0}^{m-1} \hat{P}_{2t,k}^2 \right)^{-\frac{1}{2}}$; $\phat{c} \gets \pm \left( 2 \sum_{k=0}^{m-1} \hat{P}_{2t+1,k}^2 \right)^{-\frac{1}{2}}$ \\
			\For{$k\gets0$ \KwTo $m-1$}{
				$M_{2t,k} \gets c \times \hat{P}_{2t,k}$; $M_{2t+1,k} \gets -\phat{c} \times \hat{P}_{2t+1,k}$ \\
				$M_{2t,k+m} \gets c \times \hat{P}_{2t,k}$; $M_{2t+1,k+m} \gets \phat{c} \times \hat{P}_{2t+1,k}$
			}
		}
	}
	\caption{Orthogonal Matrix Generation}
	\label{algo}
\end{algorithm}

\section{Generating Sample Orthogonal Matrices}
\label{sec:practice}

In order to practice our method in real world scenarios,
we apply the procedure to the solutions found in the classical expansions and reproduce those orthogonal matrices currently in wide use, as samples.
As described in \S\ref{sec:orthogonal}, to generate an $n \times n$ orthogonal matrix, $n = 2m$ must be an even number.
This requirement is in fact less restrictive than that of most other generating methods, where $n$ must be a power of 2.
Therefore, all the sample matrices can be generated by our method without any problem in their dimensions.

The experiments are carried out as follows.
We first determine $n = 2m$ distinct values for the targeted sample matrix.
In fact, among the $2m$ values, half of them are the opposites of the other half, thus only $m$ distinct positive values are required.
As discussed in \S\ref{ssec:coef-induct}--\ref{ssec:algo}, the entire procedure can be separated into two batches, that are,
(i) the {\em even}-numbered polynomials $P_0(x),P_2(x),\dots,P_{n-2}(x)$,
and (ii) the {\em odd}-numbered polynomials $P_1(x),P_3(x),\dots,P_{n-1}(x)$,
iteratively and respectively from the base cases $P_0(x)$ and $P_1(x)$.
In particular, we choose only the arithmetic square roots in Algorithm~\ref{algo} to simplify the results.
At the end of the section, we illustrate that, by using arbitrary distinct values,
we are also able to produce new and unique orthogonal matrices, not just the special values of those discovered matrices.\footnote{%
	We have implemented the procedures to generate the sample orthogonal matrices in:
	\begin{center}
		\tt github.com/ChanKaHou/DiscreteOrthogonalMatrices
	\end{center}
}

\subsection{$8\times8$ Discrete Cosine Transform Matrix}

To generate the $n \times n$ ($n = 8$) DCT matrix, we must first confirm the $n$ distinct values.
Since the DCTs are also closely related to the Chebyshev polynomials~\cite{ahmed1974discrete},
where the coefficients of $P_1(x)$ are the roots of the $n$-th Chebyshev polynomial $P_n(\cos(x)) = \cos(nx)$,
that is,
\begin{equation}
	\label{eq:Chebyshev}
	P_n(x) = \cos(n \arccos(x)) = 0.
\end{equation}
Solving (\ref{eq:Chebyshev}), we have the $n$ roots to be
$$
x_i = \cos\left(\frac{i+\frac{1}{2}}{n}\pi\right)\qquad i = 0,1,\dots,n-1.
$$
Also by Algorithm~\ref{algo}, line~\ref{algo:P1-coef}, we notice that
the coefficients of $P_1(x)$ are also the set of $n$ distinct values we are using to generate the matrix.
Thus for the $8 \times 8$ orthogonal matrix,
we have the $8$ values as the roots of (\ref{eq:Chebyshev}),
$\pm\cos\left(\frac{\pi}{16}\right),\pm\cos\left(\frac{3\pi}{16}\right),\pm\cos\left(\frac{5\pi}{16}\right),\pm\cos\left(\frac{7\pi}{16}\right)$.
By taking these values, Algorithm~\ref{algo} produces an $8 \times 8$ DCT matrix as shown in Appendix~\ref{appe:DCT8x8}.
In fact, the Appendix~\ref{appe:DCT8x8} can be multiplied by $64\sqrt{8}$ and rounded to Appendix~\ref{appe:integerDCT8x8},
which is the widely used DCT type 2 matrix \cite{ShaoJ08a} that has been embedded in the next generation video coding standard, Versatile Video Coding~\cite{bross2020versatile}, and its reference software VVC Test Model~(VTM).
Figure~\ref{fig:DCT} gives the corresponding plot of the first $8$ Chebyshev polynomials.
The plot of the polynomials shows some features of cosine functions.
The roots are also in the range of $[-1,+1]$, and all of them are the cosine values of the radians in an arithmetic sequence.

\begin{figure} [h]
	\centering
	\begin{tikzpicture}
		\begin{axis}[
			smooth,
			width=\linewidth,
			unit vector ratio = 1 1,
			grid=major, grid style=dashed,
			domain=-1.1:1.1, ymin=-0.6, ymax=+0.6,
			axis lines=middle, axis line style = thick,
			xtick={-cos(180/16),-cos(180*3/16),-cos(180*5/16),-cos(180*7/16),+cos(180*7/16),+cos(180*5/16),+cos(180*3/16),+cos(180/16)},
			xticklabels={$-\cos\left(\frac{\pi}{16}\right)$,$-\cos\left(\frac{3\pi}{16}\right)$,$-\cos\left(\frac{5\pi}{16}\right)$,$-\cos\left(\frac{7\pi}{16}\right)$,$\cos\left(\frac{7\pi}{16}\right)$,$\cos\left(\frac{5\pi}{16}\right)$, ,
				$\cos\left(\frac{\pi}{16}\right)$},
			x tick label style={rotate=60,anchor=east,overlay},
			ytick distance=0.1,
			legend pos=south east, legend cell align=left
			]
			\addlegendentry{$P_0\left(x\right)$} \addplot[black,thick]{1/(8^0.5)};
			\addlegendentry{$P_1\left(x\right)$} \addplot[red,thick]{0.5*(x)};
			\addlegendentry{$P_2\left(x\right)$} \addplot[green,thick]{(x^2-0.5)};
			\addlegendentry{$P_3\left(x\right)$} \addplot[blue,thick]{2*(x^3-0.75*x)};
			\addlegendentry{$P_4\left(x\right)$} \addplot[olive,thick]{4*(x^4-x^2+0.125)};
			\addlegendentry{$P_5\left(x\right)$} \addplot[magenta,thick]{8*(x^5-1.25*x^3+0.3125*x)};
			\addlegendentry{$P_6\left(x\right)$} \addplot[orange,thick]{16*(x^6-1.5*x^4+0.5625*x^2-0.03125)};
			\addlegendentry{$P_7\left(x\right)$} \addplot[cyan,thick]{32*(x^7-1.75*x^5+0.875*x^3-0.109375*x)};
		\end{axis}
	\end{tikzpicture}
	\caption{The first 8 polynomials of the DCT matrix in domain $x\in\left(-1,+1\right)$, the corresponding 8 roots as $\left\lbrace\pm\cos\left(\frac{\pi}{16}\right),\pm\cos\left(\frac{3\pi}{16}\right),\pm\cos\left(\frac{5\pi}{16}\right),\pm\cos\left(\frac{7\pi}{16}\right)\right\rbrace$.}
	\label{fig:DCT}
\end{figure}
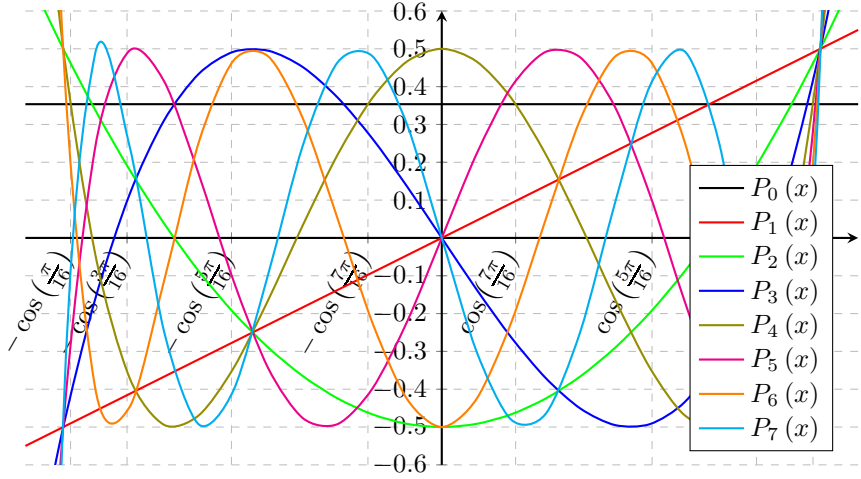

\subsection{$8\times8$ Discrete Tchebichef Transform Matrix}

The Discrete Tchebichef Transform~(DTT) is another widely used transform method by using the Chebyshev polynomials~\cite{mukundan2001image},
which has as good energy compaction properties as of the DCT,
and works better for a certain class of 2D information.
Because the Chebyshev polynomials are too complex, unlike in the DCT case,
the roots of the $n$-th polynomial $P_n(x) = 0$ are difficult to obtain for setting the values to generate the matrix.
However, as discussed in the DCT case, the discovered orthogonal matrices can help us determine the coefficients of polynomial $P_1(x)$,
thus the values for our generation method.
For example, a $4 \times 4$ DTT matrix has been discussed in~\cite{nakagaki2007fast} where the coefficients of $P_1(x)$ are
$$
-\frac{3\sqrt{5}}{10},-\frac{\sqrt{5}}{10},+\frac{\sqrt{5}}{10},+\frac{3\sqrt{5}}{10},
$$
which form an arithmetic sequence. We can use these values to generate the orthogonal matrix.
Furthermore, consider the loop on line~\ref{algo:unit-vector} of Algorithm~\ref{algo} to normalize each row to a unit vector,
the values for generating the matrix can be scaled arbitrarily.
Therefore, we can use a better distributed arithmetic sequence
$$
-\frac{3}{4}, -\frac{1}{4}, +\frac{1}{4}, +\frac{3}{4}
$$
in the range of $[-1, 1]$ as the generating values.
As a result, we obtain the same matrix as in Appendix~\ref{appe:DTT4x4} by using Algorithm~\ref{algo} with these values.
Similarly, in order to generate an $8 \times 8$ DTT matrix, we must determine the 8 generating values first.
Following the same principle, we use the evenly distributed arithmetic sequence of 8 values in the range of $[-1, 1]$,
$$
\pm\frac{1}{8},\pm\frac{3}{8},\pm\frac{5}{8},\pm\frac{7}{8}.
$$
We are able to obtain the $8 \times 8$ DTT matrix as in Appendix~\ref{appe:DTT8x8},
which is identical to the one describe in~\cite{mukundan2001image}.
Similar to Appendix~\ref{appe:integerDCT8x8},
Appendix~\ref{appe:DTT4x4} and \ref{appe:DTT8x8} can also be scaled and rounded to an integer matrix like Appendix~\ref{appe:integerDTT4x4} and \ref{appe:integerDTT8x8},
which can be used in coding applications.
Figure~\ref{fig:DTT} shows the corresponding plot of the first 8 DTT polynomials.
Different from the DCT case, the values to generate a DTT matrix themselves form an arithmetic sequence.
Thus, for the generation of a general $2m \times 2m$ DTT matrix
we should set the arithmetic sequence
$$
\pm\frac{1}{2m}, \pm\frac{3}{2m}, \dots, \pm\frac{2m-1}{2m}
$$
as the generating values in our method. This enables us to generate DTT matrices of arbitrarily large sizes.

\begin{figure} [h]
	\centering
	\begin{tikzpicture}
		\begin{axis}[
			smooth,
			width=\linewidth,
			unit vector ratio = 0.8 0.8,
			grid=major, grid style=dashed,
			domain=-1.1:1.1, ymin=-1.0, ymax=+1.0,
			axis lines=middle, axis line style = thick,
			xtick={-7/8,-5/8,-3/8,-1/8,+1/8,+3/8,+5/8,+7/8},
			xticklabels={$-\frac{7}{8}$,$-\frac{5}{8}$,$-\frac{3}{8}$,$-\frac{1}{8}$,
				$+\frac{1}{8}$,$+\frac{3}{8}$,$+\frac{5}{8}$,$+\frac{7}{8}$},
			ytick distance=0.2,
			legend pos=south east, legend cell align=left
			]
			\addlegendentry{$P_0\left(x\right)$} \addplot[black,thick]{1/(8^0.5)};
			\addlegendentry{$P_1\left(x\right)$} \addplot[red,thick]{0.6172134*(x)};
			\addlegendentry{$P_2\left(x\right)$} \addplot[green,thick]{1.2344268*(x^2-0.3281250)};
			\addlegendentry{$P_3\left(x\right)$} \addplot[blue,thick]{2.6259518*(x^3-0.5781250*x)};
			\addlegendentry{$P_4\left(x\right)$} \addplot[olive,thick]{6.0168115*(x^4-0.7991071*x^2+0.0725098)};
			\addlegendentry{$P_5\left(x\right)$} \addplot[magenta,thick]{15.3381041*(x^5-0.9895833*x^3+0.1826288*x)};
			\addlegendentry{$P_6\left(x\right)$} \addplot[orange,thick]{46.2167518*(x^6-1.1434659*x^4+0.3055975*x^2-0.0111580)};
			\addlegendentry{$P_7\left(x\right)$} \addplot[cyan,thick]{190.4421354*(x^7-1.2536058*x^5+0.4145900*x^3-0.0312727*x)};
		\end{axis}
	\end{tikzpicture}
	\caption{The first 8 polynomials of the DTT matrix in domain $x\in\left(-1,+1\right)$, the corresponding 8 roots as $\left\lbrace\pm\frac{1}{8},\pm\frac{3}{8},\pm\frac{5}{8},\pm\frac{7}{8}\right\rbrace$.}
	\label{fig:DTT}
\end{figure}
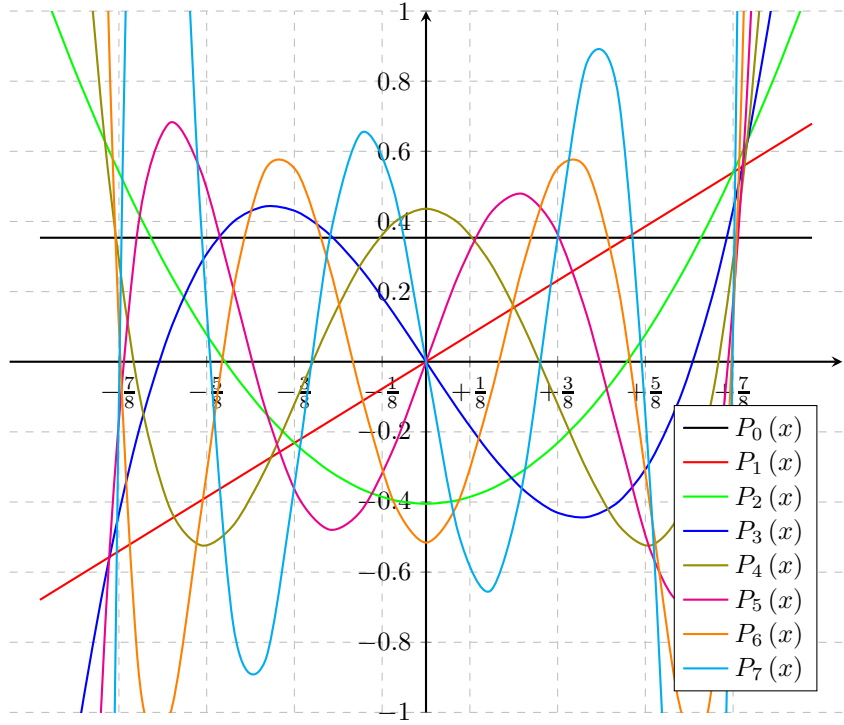

\subsection{Further Discussion}

It may not always be possible to come up with natural discretizations as in these examples.
Switching to our method, we need only to determine the generating values, then the corresponding matrix can be produced accordingly.
As indicated in the DCT and DTT cases, their generating values have certain patterns,
of which we can make use to produce larger orthogonal matrices of the same class.
In fact, our method has the advantage to accept any real numbers as the generating values,
the DCT and DTT are only two well-known cases serving as the evidence of success in our practice.
There are other potential sequences such as triangular numbers $\pm1,\pm3,\pm6,\pm10,\dots$,
prime numbers $\pm2,\pm3,\pm5,\pm7,\dots$ and Fibonacci numbers $\pm1,\pm2,\pm3,\pm5,\dots$
that can be examined further for application.
The respective $8 \times 8$ orthogonal matrices from these sequences are shown in Appendix~\ref{appe:DTriT8x8}, \ref{appe:DPrimeT8x8} and \ref{appe:DFibonacciT8x8}.
In the same way like Appendix~\ref{appe:integerDCT8x8}, all of them can be multiplied by $64\sqrt{8}$ and rounded to integer matrices for applicable applications.
It is worth noting that for polynomials of high degrees, some definitions near $[-1,1]$
classify a large percentage of the determined values as potential polynomial roots.
We can scale the polynomials for optimal root arrangement.
However, significantly decreasing the scaling factor will increase the energy of polynomials, always resulting in very large $c_{(k,0)}$.
Although this in turn affects the weight function when dealt with in the continuous form, it is out of the scope of this work.

On the other hand, there is a weak point of our method that we only present these polynomials in the form of approximate coefficients.
Although all the polynomials for the orthogonal matrices can be formulated in the accurate form like those 0 and 1-degree polynomials,
it will be too complex to read and implement for the higher degree polynomials.
When we have to approximate the coefficients iteratively, rounding errors must be taken into consideration.
Further, our method may not support the class of discrete polynomials that are orthogonal on non-uniform lattice, such as the rotation matrix in a 3D transformer.
Because a rotation matrix is not limited to $2m\times2m$ in size and its determinant must meet an additional condition, i.e., equal to $\pm1$.
In real world applications, the DCT and DTT have been widely use in image compression.
Also, our method has the potential to apply to cryptography.
For example, the user can arbitrary determine a set of values as the encryption key,
and go through our algorithm to generate a unique orthogonal matrix to map the plain text into the cipher text.
Using the inverse matrix will be able to decrypt the cipher text back to the original.

\section{Conclusion}
\label{sec:conclude}

In this paper, we present a general method for generating discrete orthogonal matrices of arbitrary even numbered sizes,
from user determined sets of distinct positive real numbers.
We give the complete induction procedure which also leads to the formal justification and the algorithm.
Our method is able to generate a class of discrete polynomials that are orthogonal on uniform lattices.
We have reproduced the well-known DCT and DTT matrices in terms of the corresponding positive values without using the continuous polynomials.
Our method provides a shortcut to the development of undiscovered orthogonal transforms for potential applications.
Invertible transformers can be generated more efficiently that are effective for sample data testing and evaluation of new ideas.
The application of this method can help eliminating the need of heavy mathematics for using certain class of orthogonal matrices.
The results of our practices shows the power and flexibility of this generating method compared with other methods for discrete orthogonal transformers.
In addition, we show that the generated matrices have the potential to facilitate other applications and analysis.

\bibliographystyle{amsrefs}
\bibliography{bib/bibliography}

\newpage
\appendices

\begin{appendices}
\def\ApdxSize{\footnotesize}

\section{DCT Matrix}

\subsection{$8 \times 8$ DCT Matrix}\label{appe:DCT8x8}
{\ApdxSize
\begin{equation*}
\begin{bmatrix*}[r]
0.3535534 & 0.3535534 & 0.3535534 & 0.3535534 & 0.3535534 & 0.3535534 & 0.3535534 & 0.3535534\\
-0.4903926 & -0.4157348 & -0.2777851 & -0.0975452 & 0.0975452 & 0.2777851 & 0.4157348 & 0.4903926\\
0.4619398 & 0.1913417 & -0.1913417 & -0.4619398 & -0.4619398 & -0.1913417 & 0.1913417 & 0.4619398\\
-0.4157348 & 0.0975452 & 0.4903926 & 0.2777851 & -0.2777851 & -0.4903926 & -0.0975452 & 0.4157348\\
0.3535534 & -0.3535534 & -0.3535534 & 0.3535534 & 0.3535534 & -0.3535534 & -0.3535534 & 0.3535534\\
-0.2777851 & 0.4903926 & -0.0975452 & -0.4157348 & 0.4157348 & 0.0975452 & -0.4903926 & 0.2777851\\
0.1913417 & -0.4619398 & 0.4619398 & -0.1913417 & -0.1913417 & 0.4619398 & -0.4619398 & 0.1913417\\
-0.0975452 & 0.2777851 & -0.4157348 & 0.4903926 & -0.4903926 & 0.4157348 & -0.2777851 & 0.0975452\\
\end{bmatrix*}
\end{equation*}
}

\subsection{$8 \times 8$ DCT Matrix in integer, namely DEFINE\_DCT2\_P8\_MATRIX}\label{appe:integerDCT8x8}

{\ApdxSize
\begin{equation*}
\text{Appendix}~\ref{appe:DCT8x8} \times64\sqrt{8}=
\begin{bmatrix*}[r]
 64 & 64 & 64 & 64 & 64 & 64 & 64 & 64\\
-89 & -75 & -50 & -18 & 18 & 50 & 75 & 89\\
 84 & 35 & -35 & -84 & -84 & -35 & 35 & 84\\
-75 & 18 & 89 & 50 & -50 & -89 & -18 & 75\\
 64 & -64 & -64 & 64 & 64 & -64 & -64 & 64\\
-50 & 89 & -18 & -75 & 75 & 18 & -89 & 50\\
 35 & -84 & 84 & -35 & -35 & 84 & -84 & 35\\
-18 & 50 & -75 & 89 & -89 & 75 & -50 & 18\\
\end{bmatrix*}
\end{equation*}
}

\section{DTT Matrix}

\subsection{$4 \times 4$ DTT Matrix}\label{appe:DTT4x4}

{\ApdxSize
\begin{equation*}
\begin{bmatrix*}[r]
0.5000000 & 0.5000000 & 0.5000000 & 0.5000000\\
-0.6708204 & -0.2236068 & 0.2236068 & 0.6708204\\
0.5000000 & -0.5000000 & -0.5000000 & 0.5000000\\
-0.2236068 & 0.6708204 & -0.6708204 & 0.2236068\\
\end{bmatrix*}
\end{equation*}
}

\subsection{$4 \times 4$ DTT Matrix in integer}\label{appe:integerDTT4x4}

{\ApdxSize
\begin{equation*}
\text{Appendix}~\ref{appe:DTT4x4} \times128=
\begin{bmatrix*}[r]
64 & 64 & 64 & 64\\
-86 & -29 & 29 & 86\\
64 & -64 & -64 & 64\\
-29 & 86 & -86 & 29\\
\end{bmatrix*}
\end{equation*}
}

\subsection{$8 \times 8$ DTT Matrix}\label{appe:DTT8x8}

{\ApdxSize
\begin{equation*}
\begin{bmatrix*}[r]
0.3535534 & 0.3535534 & 0.3535534 & 0.3535534 & 0.3535534 & 0.3535534 & 0.3535534 & 0.3535534\\
-0.5400617 & -0.3857584 & -0.2314550 & -0.0771517 & 0.0771517 & 0.2314550 & 0.3857584 & 0.5400617\\
0.5400617 & 0.0771517 & -0.2314550 & -0.3857584 & -0.3857584 & -0.2314550 & 0.0771517 & 0.5400617\\
-0.4308202 & 0.3077287 & 0.4308202 & 0.1846372 & -0.1846372 & -0.4308202 & -0.3077287 & 0.4308202\\
0.2820380 & -0.5237849 & -0.1208734 & 0.3626203 & 0.3626203 & -0.1208734 & -0.5237849 & 0.2820380\\
-0.1497862 & 0.4921546 & -0.3637664 & -0.3209704 & 0.3209704 & 0.3637664 & -0.4921546 & 0.1497862\\
0.0615457 & -0.3077287 & 0.5539117 & -0.3077287 & -0.3077287 & 0.5539117 & -0.3077287 & 0.0615457\\
-0.0170697 & 0.1194880 & -0.3584641 & 0.5974401 & -0.5974401 & 0.3584641 & -0.1194880 & 0.0170697\\
\end{bmatrix*}
\end{equation*}
}

\subsection{$8 \times 8$ DTT Matrix in integer}\label{appe:integerDTT8x8}

{\ApdxSize
\begin{equation*}
\text{Appendix}~\ref{appe:DTT8x8} \times64\sqrt{8}=
\begin{bmatrix*}[r]
64 & 64 & 64 & 64 & 64 & 64 & 64 & 64\\
-98 & -70 & -42 & -14 & 14 & 42 & 70 & 98\\
98 & 14 & -42 & -70 & -70 & -42 & 14 & 98\\
-78 & 56 & 78 & 33 & -33 & -78 & -56 & 78\\
51 & -95 & -22 & 66 & 66 & -22 & -95 & 51\\
-27 & 89 & -66 & -58 & 58 & 66 & -89 & 27\\
11 & -56 & 100 & -56 & -56 & 100 & -56 & 11\\
-3 & 22 & -65 & 108 & -108 & 65 & -22 & 3\\
\end{bmatrix*}
\end{equation*}
}

\section{$8 \times 8$ Discrete Triangular Matrix}\label{appe:DTriT8x8}

{\ApdxSize
\begin{equation*}
\begin{bmatrix*}[r]
0.3535534 & 0.3535534 & 0.3535534 & 0.3535534 & 0.3535534 & 0.3535534 & 0.3535534 & 0.3535534\\
-0.5852057 & -0.3511234 & -0.1755617 & -0.0585206 & 0.0585206 & 0.1755617 & 0.3511234 & 0.5852057\\
0.5773204 & -0.0045458 & -0.2500207 & -0.3227539 & -0.3227539 & -0.2500207 & -0.0045458 & 0.5773204\\
-0.3892916 & 0.4438069 & 0.3647887 & 0.1357083 & -0.1357083 & -0.3647887 & -0.4438069 & 0.3892916\\
0.2033226 & -0.5849516 & 0.0430456 & 0.3385833 & 0.3385833 & 0.0430456 & -0.5849516 & 0.2033226\\
-0.0773104 & 0.4211527 & -0.4960307 & -0.2657199 & 0.2657199 & 0.4960307 & -0.4211527 & 0.0773104\\
0.0190005 & -0.1811381 & 0.5573480 & -0.3952104 & -0.3952104 & 0.5573480 & -0.1811381 & 0.0190005\\
-0.0030692 & 0.0487665 & -0.3001014 & 0.6383976 & -0.6383976 & 0.3001014 & -0.0487665 & 0.0030692\\
\end{bmatrix*}
\end{equation*}
}

\section{$8 \times 8$ Discrete Prime Matrix}\label{appe:DPrimeT8x8}

{\ApdxSize
\begin{equation*}
\begin{bmatrix*}[r]
0.3535534 & 0.3535534 & 0.3535534 & 0.3535534 & 0.3535534 & 0.3535534 & 0.3535534 & 0.3535534\\
-0.5306686 & -0.3790490 & -0.2274294 & -0.1516196 & 0.1516196 & 0.2274294 & 0.3790490 & 0.5306686\\
0.5492456 & 0.0655064 & -0.2569865 & -0.3577655 & -0.3577655 & -0.2569865 & 0.0655064 & 0.5492456\\
-0.4376551 & 0.2599606 & 0.3850049 & 0.3043842 & -0.3043842 & -0.3850049 & -0.2599606 & 0.4376551\\
0.2676517 & -0.5675875 & -0.0249981 & 0.3249339 & 0.3249339 & -0.0249981 & -0.5675875 & 0.2676517\\
-0.1631185 & 0.5245744 & -0.2465314 & -0.3707240 & 0.3707240 & 0.2465314 & -0.5245744 & 0.1631185\\
0.0411317 & -0.2203482 & 0.5552774 & -0.3760609 & -0.3760609 & 0.5552774 & -0.2203482 & 0.0411317\\
-0.0155286 & 0.1164647 & -0.4891517 & 0.4969160 & -0.4969160 & 0.4891517 & -0.1164647 & 0.0155286\\
\end{bmatrix*}
\end{equation*}
}

\section{$8 \times 8$ Discrete Fibonacci Matrix}\label{appe:DFibonacciT8x8}

{\ApdxSize
\begin{equation*}
\begin{bmatrix*}[r]
0.3535534 & 0.3535534 & 0.3535534 & 0.3535534 & 0.3535534 & 0.3535534 & 0.3535534 & 0.3535534\\
-0.5661385 & -0.3396831 & -0.2264554 & -0.1132277 & 0.1132277 & 0.2264554 & 0.3396831 & 0.5661385\\
0.5824600 & -0.0286456 & -0.2196161 & -0.3341984 & -0.3341984 & -0.2196161 & -0.0286456 & 0.5824600\\
-0.4160039 & 0.3684606 & 0.3744035 & 0.2258307 & -0.2258307 & -0.3744035 & -0.3684606 & 0.4160039\\
0.1877387 & -0.5437577 & -0.0518894 & 0.4079084 & 0.4079084 & -0.0518894 & -0.5437577 & 0.1877387\\
-0.0798432 & 0.4748570 & -0.3025637 & -0.4202274 & 0.4202274 & 0.3025637 & -0.4748570 & 0.0798432\\
0.0222374 & -0.2801910 & 0.5692770 & -0.3113233 & -0.3113233 & 0.5692770 & -0.2801910 & 0.0222374\\
-0.0072786 & 0.1528496 & -0.4658274 & 0.5094987 & -0.5094987 & 0.4658274 & -0.1528496 & 0.0072786\\
\end{bmatrix*}
\end{equation*}
}

\end{appendices}

\end{document}